\documentclass[11pt]{article}

\usepackage{graphicx,amssymb,amsmath}
\usepackage[numbers]{natbib}
\usepackage{lscape}
\usepackage{graphics}
\usepackage{graphicx}
\usepackage{picins}
\usepackage{color}
\usepackage{srcltx}
\usepackage{bbm}
\newcommand{\E}{\mathbb{E}}

\newcommand{\halmos}{\mbox{} \hfill $\Box$}

\newcommand{\be}{\begin{equation}}
\newcommand{\ee}{\end{equation}}

\newtheorem{proposition}{Proposition}
\newtheorem{corollary}{Corollary}
\newtheorem{lemma}{Lemma}
\newtheorem{theorem}{Theorem}
\newenvironment{proof}{\paragraph*{\it Proof.}}{\halmos}

\begin{document}

\author{
Sylvain Delattre \\ LPMA, Universit\'e Paris Diderot (Paris 7)\\ delattre@math.univ-paris-diderot.fr\\
$~$\\
Christian Y. Robert \\ SAF, Universit\'e Lyon 1 \\ christian.robert@univ-lyon1.fr\\
$~$\\
Mathieu Rosenbaum\\ LPMA, Universit\'e Pierre et Marie Curie (Paris 6)\\ mathieu.rosenbaum@upmc.fr 
 }
 \title{Estimating the efficient price from the order flow:\\
  a Brownian Cox process approach}

\maketitle
\begin{abstract}
At the ultra high frequency level, the notion of price of an asset is very ambiguous.
Indeed, many different prices can be defined (last traded price, best bid price, mid price,\ldots).
Thus, in practice, market participants face the problem of choosing a price when implementing their strategies. 
In this work, we propose a notion of efficient price which seems relevant in practice. Furthermore, we provide
a statistical methodology enabling to estimate this price form the order flow.
\end{abstract}



\section{Introduction}

\subsection{What is the high frequency price ?}

The classical approach of mathematical finance is to consider that the prices of basic products (future, stock,\ldots)
are observed on the market. In particular, their values are used in order to price complex derivatives. Since options traders typically rebalance their portfolio once or a few times a day, such derivatives pricing problems typically occur at the daily scale. 
 
When working at the ultra high frequency scale, even pricing a basic product, that is assigning a price to it, becomes a challenging issue. Indeed, one has access to trades and quotes in the order book so that at a given time, many different notions of price can be defined for the same asset: last traded price, best bid price, best ask price, mid price, volume weighted average price,\ldots
This multiplicity of prices is problematic for many market participants. For example, market making strategies or brokers optimal execution algorithms often require single prices of plain assets as inputs. 

Choosing one definition or another for the price can sometimes lead to very significantly different outcomes for the strategies. This is for example the case when the tick value (the minimum price increment allowed on the market) is rather large. Indeed, this implies that the 
prices mentioned above differ in a non negligible way. 

In practice, high frequency market participants are not looking for the ``fair" economic value of the asset. What they need is rather a price whose value at some given time summarizes in a suitable way the opinions of market participants at this time. This price is called {\it efficient price}. Hence, this paper aims at providing a statistical procedure in order to estimate this efficient price.  

\subsection{Ideas for the estimation strategy}

In this paper, we focus on the case of large tick assets. We define them as assets for which the bid-ask spread is almost always 
equal to one tick. Our goal is then to infer an efficient price for this type of asset. Naturally, it is reasonable to assume that the efficient price essentially lies inside the bid-ask spread but we wish to say more. 

In order to retrieve the efficient price, the classical approach is to consider the imbalance of the order book, that is the difference between the available volumes at the best bid and best ask levels, see for example \cite{CL12}. Indeed, it is often said by market participants that ``the price is where the volume is not". Here we consider a dynamic version of this idea through the information available in the {\it order flow}. More precisely, we assume that the intensity of arrival of the limit order flow at the best bid or the best ask level depends on the distance between the efficient price and the considered level: if this distance is large, the intensity should be high and conversely. Thus, we assume the intensity can be written as an increasing deterministic function of this distance. This function is called the {\it order flow response function}. In our approach, a crucial step is to estimate the response function in a non parametric way. Then, this functional estimator is used in order to retrieve the efficient price.  

Note that it is also possible to use the buy or sell market order flow. In that case, the intensity of the flow should be high when the distance is small. Indeed, in this situation, market takers are not loosing too much money (with respect to the efficient price) when crossing the spread. 

\subsection{Organization of the paper}

The paper is organized as follows. The model and the assumptions are described in Section \ref{mod}. Particular properties of the
efficient price are given in Section \ref{effpr} and the main statistical procedure is explained in Section \ref{stat}. The theorems about the response function can be found in Section \ref{thh} and the
limiting behavior of the estimator of the efficient price is given in Section \ref{theffpr}. One numerical illustration can be found in Section \ref{num} and a conclusion is given in Section \ref{conclu}. Finally the proofs are relegated to Section \ref{proofs}.

\section{The model}\label{mod}

\subsection{Description of the model}

We assume the tick size is equal to one, meaning that the asset can only take integer values. Moreover, the efficient price is given by $$P_{t}=P_{0}+\sigma W_{t},~t\in[0,T],$$
where $( W_{t}) _{t\geq 0}$ is a Brownian motion on some filtered probability
space $(\Omega ,\mathcal{F},(\mathcal{F}_{t})_{t\geq 0},\mathbb{P})$ and $P_{0}$ is a $\mathcal{F}_{0}$-measurable random variable, independent of $(
W_{t}) _{t\geq 0}$
and uniformly distributed on $[p_{0},p_{0}+1)$, with $p_{0}\in \mathbb{N}$. Note that such a simple dynamics for the efficient price is probably still reasonable at our high frequency scale. 

Let $Y_{t}$ be the fractional part of $P_t$, denoted by $\{P_t\}$, that is: 
$$Y_t=\{P_t\}=P_{t}-\lfloor P_{t}\rfloor.$$ To fix ideas and without loss of generality, we focus in the rest of the paper on the limit order flow at the best bid level. We assume that when a limit order is posted at time $t$ at the best bid level, its price $B_t$ is given by $\lfloor P_{t}\rfloor$. Therefore at time $t$, the efficient price is
$$P_t=B_t+Y_t.$$

We denote by $N_{t}$ the total number of limit orders posted over $[0,t]$. In order to translate the fact that the intensity of $N_t$ should be an increasing function of $Y_t$,
we assume that $( N_{t}) _{t\geq 0}$ is a
Cox process (also known as doubly stochastic Poisson process) with arrival intensity at time $t$ given
by $$\mu h( Y_{t}),$$ where $\mu $ is a positive constant and where the response function $h:[0,1)\rightarrow \mathbb{R}_{+}$ satisfies the following assumption:
\paragraph{Assumption H1: Response function}
The function $h$ is $\mathcal{C}^1$ on $[0,1)$, with derivative $h'$ for which there exist some positive constants 
$c_1$ and $c_2$ such that for any $x$ in $[0,1)$,
$$c_1\leq h'(x)\leq c_2.$$ Moreover, we have the following identifiability condition:
\begin{equation*}
\int_{0}^{1}h(x)dx=1.
\end{equation*}

Remark that H1 implies in particular that $h$ is bounded on $[0,1)$. The response function expresses
the influence of the efficient price on the order flow. In particular, the limiting case where $h$ is constant corresponds
to orders arriving according to a standard Poisson process. This simple arrival mechanism is often assumed in practice for
technical convenience, see for example \cite{CL12}, \cite{far}. However, it is not very realistic. Nevertheless, it is clear
that our model is still probably too simple, $Y_t$ being the only source of fluctuation of the intensity. For example, a more reasonable version should probably allows for seasonalities in the flow and the possible influence of exogenous variables.

We assume that we observe the point process $(N_{t}) $ on $[0,T]$. To get asymptotic
properties, we let $T$ tend to infinity. It will be also necessary to
assume that $\mu =\mu _{T}$ depends on $T$. More precisely, we have the following assumption: 

\paragraph{Assumption H2: Asymptotic setting}
For some $\varepsilon>0$, as $T$ tends to infinity, $$T^{5/2+\varepsilon}/\mu_T\rightarrow 0.$$

Note that up to scaling modifications, we could also consider a setting where $T$ is fixed and the tick size is not constant equal to $1$ but tends to zero.

\section{Properties of the process $Y_t$}\label{effpr}

The intensity of the order flow is $\mu_T h(Y_t)$. This process $(Y_t)_{t\geq 0}$ has several nice properties. 
\subsection{Markov property}
First, 
recall that if $U$ is uniformly distributed on $[0,1]$ and $X$ is 
a real-valued random variable, which is independent of $U$ then $\{U + X\}$ is also uniformly distributed on $[0,1]$.
Thus, since $P_0$ is uniformly distributed on $[0,1]$, we obtain that $(Y_t)$ is a stationary Markov process. Then, from the properties
of the sample paths of the Brownian motion, it is clear that $(Y_t)$ is Harris recurrent and therefore, for $f$ a bounded positive measurable function, almost surely,
$$\underset{T\rightarrow +\infty}{\text{lim}}\frac{1}{T}\int_0^Tf(Y_s)ds=\int_0^1f(s)ds,$$ 
see for example Theorem 3.12 in \cite{ry}.
\subsection{Regenerative property}\label{reg}
Beyond the Markov property, the process $(Y_t)$ also enjoys a regenerative property.
We define the sequence of stopping time $(\nu _{n}) _{n\in 
\mathbb{N}}$ the following way: $\nu _{0}=0$, $\nu _{1}=\inf\left\{
t>0:P_{t}\in \mathbb{N}\right\} $ and for $n\geq 2$:
\begin{equation}\label{stop}
\nu _{n}=\inf \{ t>\nu _{n-1}:P_{t}=P_{\nu _{n-1}}\pm
1\} =\inf\{ t>\nu _{n-1}:W_{t}=W_{\nu _{n-1}}\pm 1/\sigma\}
.
\end{equation}%
The cycles $(Y_{t+\nu _{n}})_{0\leq t<\nu _{n+1}-\nu_{n}}$ are independent and identically distributed for $n\geq 1$. Thus $%
(Y_{t})$ is a regenerative process. Note that $\nu_2-\nu_1$ has the same law
as $$\tau _{1}=\inf\{ t\geq 0,|W_{t}|=1/\sigma\}.$$ The Laplace transform of $\tau_1$ is given by
\begin{equation*}
\mathbb{E}[e^{-\gamma\tau _{1}}] =1/\cosh\big((\sqrt{2\gamma})/\sigma\big),~~\gamma\geq 0.
\end{equation*}%
From this expression, we get $\mathbb{E}[\tau_1]=1/\sigma^2$ and $\mathbb{E}[\tau_1^2]=5/(3\sigma^4)$. Thus for $f$ a bounded positive measurable function,
by Theorem VI.3.1\ in \cite{As03}, we get that almost surely
$$\underset{T\rightarrow +\infty}{\text{lim}}\frac{1}{T}\int_0^Tf(Y_s)ds=\frac{1}{\mathbb{E}[\nu _{2}-\nu _{1}]}\mathbb{E}\big[\int_{\nu
_{1}}^{\nu _{2}}f(Y_{t})dt\big] =\sigma^2\mathbb{E}\big(\int_{0}^{\tau
_{1}}f(\{\sigma W_{t}\}) dt\big).$$ 
In particular, this implies that
$$\sigma^2\mathbb{E}\big[\int_{0}^{\tau_{1}}f(\{\sigma W_{t}\}) dt\big]=\int_0^1f(s)ds.$$
Furthermore, 
by Theorem VI.3.2, we get the following convergence in distribution as $T\rightarrow \infty$:
\begin{equation}\label{CLT1}
\sqrt{T}\Big(\frac{1}{T}\int_{0}^{T}f(Y_{t})dt-\sigma^2\mathbb{E}%
\big[\int_{0}^{\tau _{1}}f(\{\sigma W_{t}\}) dt\big]
\Big)\overset{d}{\rightarrow }N( 0,\sigma^2\text{Var}[Z^f]),
\end{equation}
with $Z^f$ a centered random variable defined by
$$Z^f=\int_{0}^{\tau_1}f(\{\sigma W_t\})dt-\tau_1\sigma^2\mathbb{E}%
\big[\int_{0}^{\tau _{1}}f(\{\sigma W_{t}\}) dt\big].$$

We end this section with two particular cases for the function $f$ which will be useful in the following.
First, if $f=h$, using the identifiability condition in Assumption H1, we get that almost surely:
\begin{equation*}
\underset{T\rightarrow +\infty}{\text{lim}}\frac{1}{T}\int_{0}^{T}h(Y_{t})dt=\sigma^2\mathbb{E}\big[\int_{0}^{\tau_{1}}h(\{\sigma W_{t}\}) dt\big]=1.
\end{equation*}
Therefore,
\begin{eqnarray*}
Z^h&=&\int_{0}^{\tau _{1}}\big(\mathbb{I}_{\{W_{t}<0\} }(
h( 1+\sigma W_{t}) -1) +\mathbb{I}_{\{ W_{t}>0\}
}(h( \sigma W_{t}) -1)\big) dt \\
&=&\int_{-1/\sigma}^{1/\sigma}\big(\mathbb{I}_{\{ u<0\} }( h(
1+\sigma u) -1) +\mathbb{I}_{\{ u>0\} }(h(\sigma u) -1)%
)L_{-1/\sigma,1/\sigma}( u) du,
\end{eqnarray*}%
where $L_{-1/\sigma,1/\sigma}$ is the local time of $( W_{t}) _{t\geq
0}$ stopped at $\tau _{1}$. 

In the same way, for $f(x)=\mathbb{I}_{\{h(x)\leq t\}}$, 
$t\in[0,h(1^{-}))$, we have 
$$\int_{0}^1\mathbb{I}_{\{h(s)\leq t\}}ds=h^{-1}(t).$$
So we define $Z_e(t)=Z^f$ by
\begin{eqnarray*}
Z_e(t)&=&\int_{0}^{\tau _{1}}\big( \mathbb{I}_{\{
W_{s}<0,h(1+\sigma W_{s})\leq t \} }+\mathbb{I}_{\{
W_{s}>0,h(\sigma W_{s})\leq t\} }-h^{-1}(t) \big) ds
\\
&=&\int_{-1/\sigma}^{1/\sigma}\big(\mathbb{I}_{\{ u<0,h(1+\sigma u)\leq t\} }+%
\mathbb{I}_{\{ u>0,h(\sigma u)\leq t\} }-h^{-1}( t) \big)L_{-1/\sigma,1/\sigma}( u) du.
\end{eqnarray*}%
In fact, it will be useful to see $(Z_e(t))_{t\in[0,h(1^{-}))}$ as a process indexed by $t$. Thus, we introduce the covariance function 
\begin{equation}
\rho( t_{1},t_{2}) =\text{Cov}[Z_e(t
_{1}),Z_e(t _{2})] =\mathbb{E}[Z_e(t
_{1})Z_e(t _{2})].  \label{defrho}
\end{equation}%
Note that this function and $\text{Var}[Z^h]$ can be computed as multiple integrals on $[-1/\sigma,1/\sigma]^{2}$. Indeed, for $(u,v)\in[-1/\sigma,1/\sigma]^{2}$ an explicit 
expression for the bivariate Laplace transform of $$\big(L_{-1/\sigma,1/\sigma}(u),L_{-1/\sigma,1/\sigma}(v)\big)$$ is available, see Formula (I.3.18.1) in \cite{BS}. Using differentiation, from this expression, one can easily derive\footnote{We skip the explicit writing of $\mathbb{E}[L_{-1/\sigma,1/\sigma}(u)L_{-1/\sigma,1/\sigma}(v)]$, the expression being quite heavy.} $\mathbb{E}[L_{-1/\sigma,1/\sigma}(u)L_{-1/\sigma,1/\sigma}(v)]$ which enables to deduce $\rho( t_{1},t_{2})$ and $\text{Var}[Z^h]$. 

\section{Estimating the flow response function and its inverse}\label{stat}

We want to estimate the flow response function. 
Before estimating $h$, we need to estimate $\mu _{T}$. Since%
\begin{equation*}
\mathbb{E}\big[\frac{N_{T}}{\mu _{T}T}\big]=\mathbb{E}\big[\frac{1}{T}%
\int_{0}^{T}h( Y_{t}) dt\big] =\frac{1}{T}\int_{0}^{T}\mathbb{E}%
[h( Y_{t})]dt=1,
\end{equation*}%
it is natural to propose 
\begin{equation}\label{estimmu}
\hat{\mu}_{T}=\frac{N_{T}}{T}
\end{equation}%
as an estimator for $\mu _{T}$. We have the following proposition, whose proof is given in Section \ref{Proofmuchap}.

\begin{proposition}
\label{muchap}If $\mu _{T}\rightarrow \infty $ as $T\rightarrow \infty $, we
have%
\begin{equation*}
\sqrt{T}\big( \frac{\hat{\mu}_{T}}{\mu _{T}}-1\big) \overset{d}{
\rightarrow }N(0,\sigma^2\emph{Var}[Z^h]) .
\end{equation*}
\end{proposition}

We now explain how the estimator of the flow response function is
built. Let $k_{T}$ be a known deterministic sequence of positive integers. Then define for $%
j=1,\ldots ,k_{T}$%
\begin{equation*}
\hat{\theta}_{j}=k_{T}\frac{N_{jT/k_{T}}-N_{(j-1)T/k_{T}}}{\hat{\mu}_{T}T}=%
\frac{k_{T}}{N_{T}}( N_{jT/k_{T}}-N_{(j-1)T/k_{T}}).
\end{equation*}%
Now remark that $\hat{\theta}_{j}$ is approximately equal to
$$\frac{1}{\mu _{T}T/k_{T}}\sum_{i=1}^{\lfloor \mu _{T}T/k_{T}\rfloor}\big(N_{(j-1)T/k_{T}+i/\mu_T}-N_{(j-1)T/k_{T}+(i-1)/\mu_T}\big).$$ 
Conditional on the path of $(Y_t)$, the variables in the sum are independent
and if $T/k_{T}$ is small enough, they approximately follow 
a Poisson law with parameter $\mu _{T}h(Y_{(j-1)/k_{T}})/k_T$. Therefore, if moreover $\mu _{T}T/k_{T}$ is sufficiently large, one can expect that $%
\hat{\theta}_{j}$ is close to $h(Y_{(j-1)T/k_{T}})$. In fact, throughout the paper, we assume that $k_T$ is chosen so that for some $p>0$, as $T$ tends to infinity, $$\frac{T^{p+1/2}}{k_T^{p/2}}\rightarrow 0,~~\frac{k_TT^{1/2}}{\mu_T}\rightarrow 0.$$
Note that since $T^{5/2+\varepsilon}/\mu_T\rightarrow 0$ such a sequence $k_T$ exists.

The $\hat{\theta}_{j}$ introduced above are
$k_{T}$ estimators of quantities of the form $h(u_{j})$, $u_{j}\in \lbrack
0,1]$. However, we do not have access to the values of the $u_{i}$. Nevertheless, we know that
they are uniformly distributed on $[0,1]$. We therefore rank
the $\hat{\theta}_{j}$: $\hat{\theta}_{\left( 1\right) }\leq \hat{\theta}%
_{\left( 2\right) }\leq \ldots \leq \hat{\theta}_{\left( k_{T}\right) }$. For $u\in \lbrack 0,1)$, we
define the estimator of $h(u)$  the following way:%
\begin{equation*}
\hat{h}(u)=\hat{\theta}_{(\lfloor uk_{T}\rfloor+1)
}.
\end{equation*}%
Note that the identifiability condition holds since%
\begin{equation*}
\int_{0}^{1}\hat{h}(u)du=\frac{1}{k_{T}}\sum_{j=1}^{k_{T}}\hat{\theta}_{j}=1.
\end{equation*}

Then, the estimator of $h^{-1}$ is naturally defined by:
\begin{equation}\label{hinv}
\hat{h}^{-1}(t)=\frac{1}{k_{T}}\sum_{j=1}^{k_{T}}\mathbb{I}_{\{ 
\hat{\theta}_{j}\leq t\}}.
\end{equation}
Indeed, $\hat{h}$ is the right continuous generalized inverse of $\hat{h}^{-1}$.
\section{Limit theorems for the response function}\label{thh}

We now give the theorems associated to $\hat{h}^{-1}$ and $\hat{h}$.
For $b>0$, we write $D[0,b)$ (resp. $D[0,b]$), for the space of c\`{a}dl\`{a}g functions 
from $[0,b)$ (resp. $[0,b]$) into $\mathbb{R}$. We define
the convergence in law in $D[0,b)$ as the convergence in law of the restrictions of the stochastic processes to any compact $
[0,a] $, $0<a<b$, for the Skorohod topology $J_{1}$. We have the following result for $\hat{h}^{-1}$:

\begin{theorem}\label{thhinv}
Under H1 and H2, as $T$ tends to infinity, we have
\begin{equation*}
\sqrt{T}\big( \hat{h}^{-1}( \cdot) -h^{-1}(\cdot)
\big)\overset{d}{
\rightarrow } \sigma G(\cdot )-\frac{(\cdot)}{h^{\prime }\big( h^{-1}(\cdot )\big) }%
\int_{0}^{h(1^{-})}\sigma G(v) dv,
\end{equation*}%
in $D[0,h(1^{-}))$, where $G(\cdot )$ is a continuous centered Gaussian process with covariance
function $\rho $ defined by \eqref{defrho}.
\end{theorem}
Note that although there are $k_T$ terms in the sum defining $h^{-1}$ in \eqref{hinv}, the rate of convergence in Theorem \ref{thhinv} is $\sqrt{T}$
which is slower than $\sqrt{k_T}$. This is due to the strong dependence between the $\hat{\theta}_{j}$ within each cycle and the fact that the number of cycles is of order $T$.

The same type of result holds also for $\hat{h}$:
\begin{theorem}
\label{TCL}Under H1 and H2, as $T$ tends to infinity, we have
\begin{equation*}
\sqrt{T}\big( \hat{h}( \cdot) -h( \cdot)\big)
\overset{d}{
\rightarrow } -h^{\prime}(\cdot)\sigma G\big(h(\cdot)\big)+%
h(\cdot)\int_{0}^{h(1^{-})}\sigma G(v) dv,
\end{equation*}%
in $D[0,1)$.
\end{theorem}

Note that both in Theorem \ref{thhinv} and Theorem \ref{TCL}, the covariance functions of the limiting
processes can be computed explicitly, see Section \ref{reg}.
 
\section{Estimation of the efficient price}\label{theffpr}

Let $t\in (0,T]$. A nice corollary of Theorem \ref{thhinv} is the possibility of estimating $Y_t$, which is the fractional part of the efficient price at time $t$, and therefore the efficient price itself (provided $t$ is an order arrival time). Indeed, we have an estimator $\hat{h}^{-1}$ of $h^{-1}$ and we know how to estimate $h(Y_{t})$, namely we take\footnote{Of course we can take any other time interval of length $T/k_T$ containing $t$. Also, one could probably consider a smoothed version of $\widehat{h(Y_t)}$ using a kernel rather than a rectangle window.} 
$$\widehat{h(Y_t)}=k_{T}\frac{N_{t}-N_{t-T/k_{T}}}{\hat{\mu}_{T}T}.$$
Therefore, it is natural to define an estimator $\widehat{(Y_t)}$ of $Y_t$ by
$$\widehat{Y_t}=\hat{h}^{-1}\big(\widehat{h(Y_t)}\big).$$
The following corollary shows that this estimator is indeed a good estimator of the efficient price $Y_t$.
\begin{corollary}\label{effestim} For $t>0$, under H1 and H2, as $T$ tends to infinity, we have 
$$\sqrt{T}\big(\widehat{Y_t} -Y_t\big)\overset{d}{
\rightarrow}\sigma G\big(h(Y_t)\big),$$
with $G$ independent of $Y_t$.
\end{corollary}
Thus, thanks to this approach it is possible to retrieve the efficient price from the order flow, with the usual parametric rate $\sqrt{T}$. Remark that the law of $\sigma G\big(h(Y_t)\big)$ can be explicitly computed. Indeed, conditional on $Y_t$ we have a centered Gaussian random variable with variance $\sigma^2\rho\big(h(Y_t),h(Y_t)\big)$. Then we can use the fact that $Y_t$ is independent of $G$ and uniformly distributed.

\section{One numerical illustration}\label{num}

In order to give an idea of the way our procedure behaves, we give in this section one short numerical illustration.
We take $\sigma=1$, $T=5$, $\mu_T=1000$ and $k_T=150$. We consider two different cases\footnote{Note that strictly speaking, in the second case, the function $h$ does not fully satisfy our assumptions since the derivative in zero is equal to zero.} for the function $h$: $h(u)=2u$ and $h(u)=4u^3$.
We give in Figure 1 examples of estimator of the order flow response function $h$, over one sample path in each case. Averages of estimators together with 95$\%$ confidence intervals obtained from 10 000 Monte Carlo simulations are drawn in Figure 2.

\begin{center}
\begin{tabular}{cc}\label{hchapeau}
\includegraphics[height=7cm,width=6.5cm]{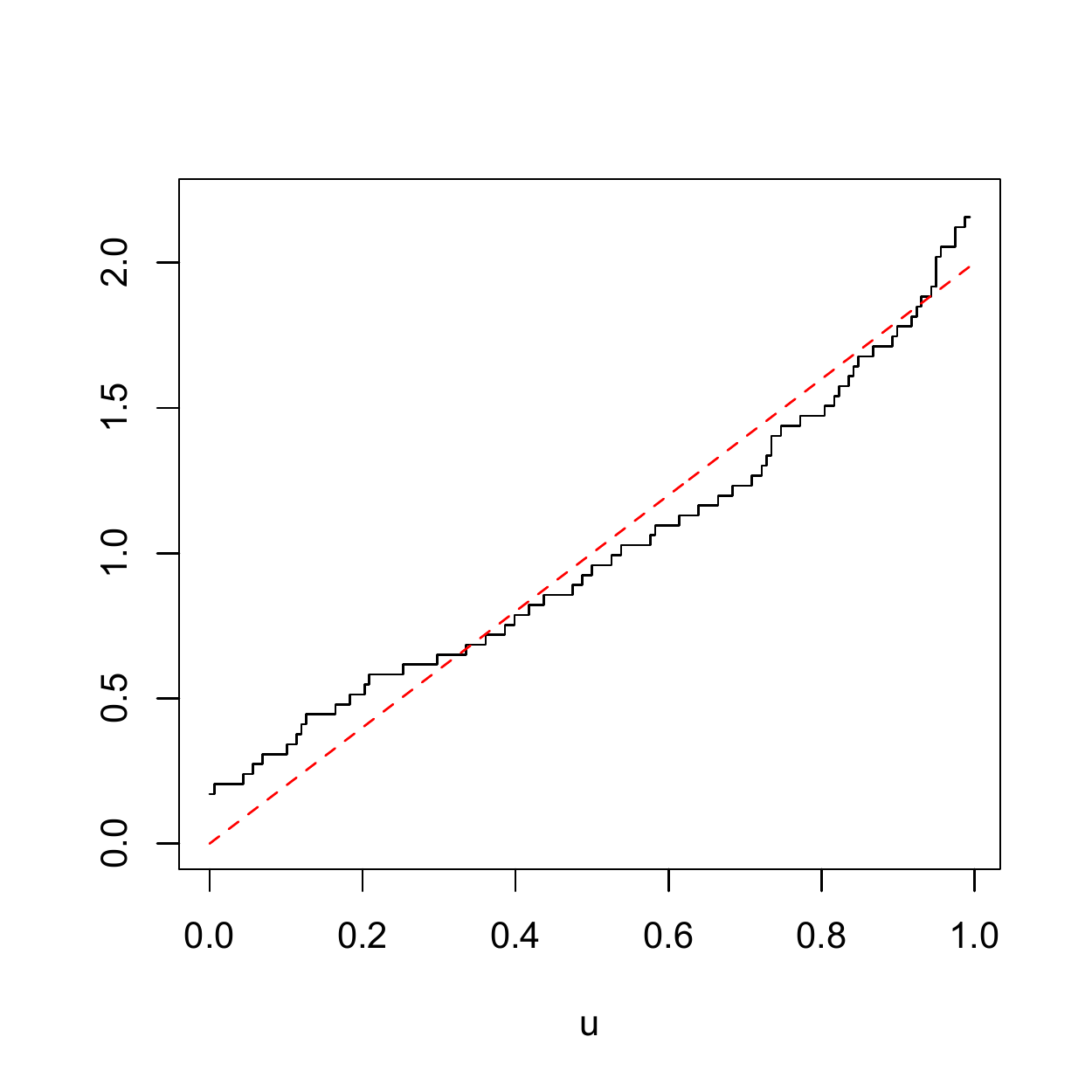} &
\includegraphics[height=7cm,width=6.5cm]{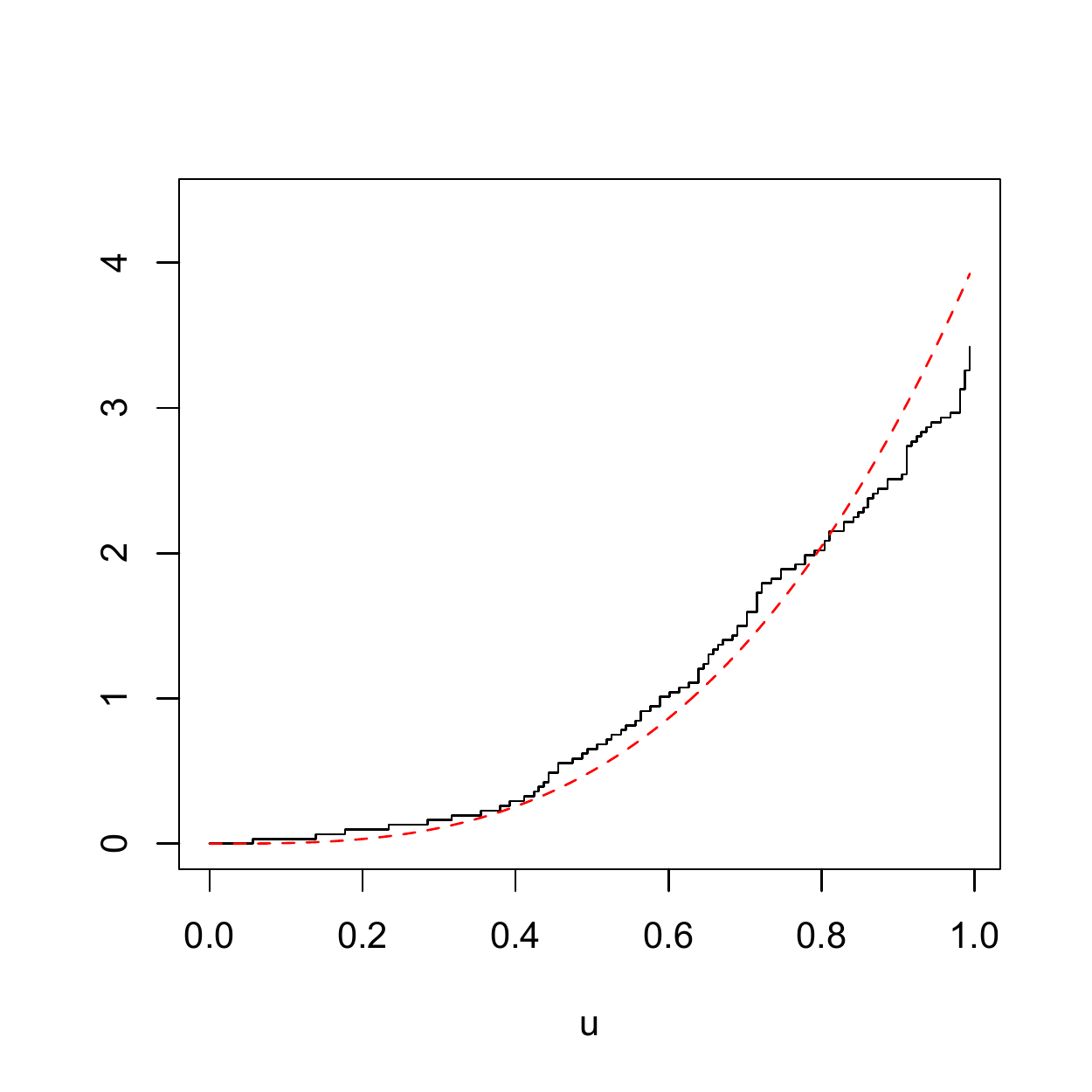}
\end{tabular}
\caption{\textit{One sample path for $h(u)=2u$ (left) and $h(u)=4u^3$ (right). In red the function $h$ and in black its estimation.}}
\begin{tabular}{cc}\label{hchapeaumc}
\includegraphics[height=7cm,width=6.5cm]{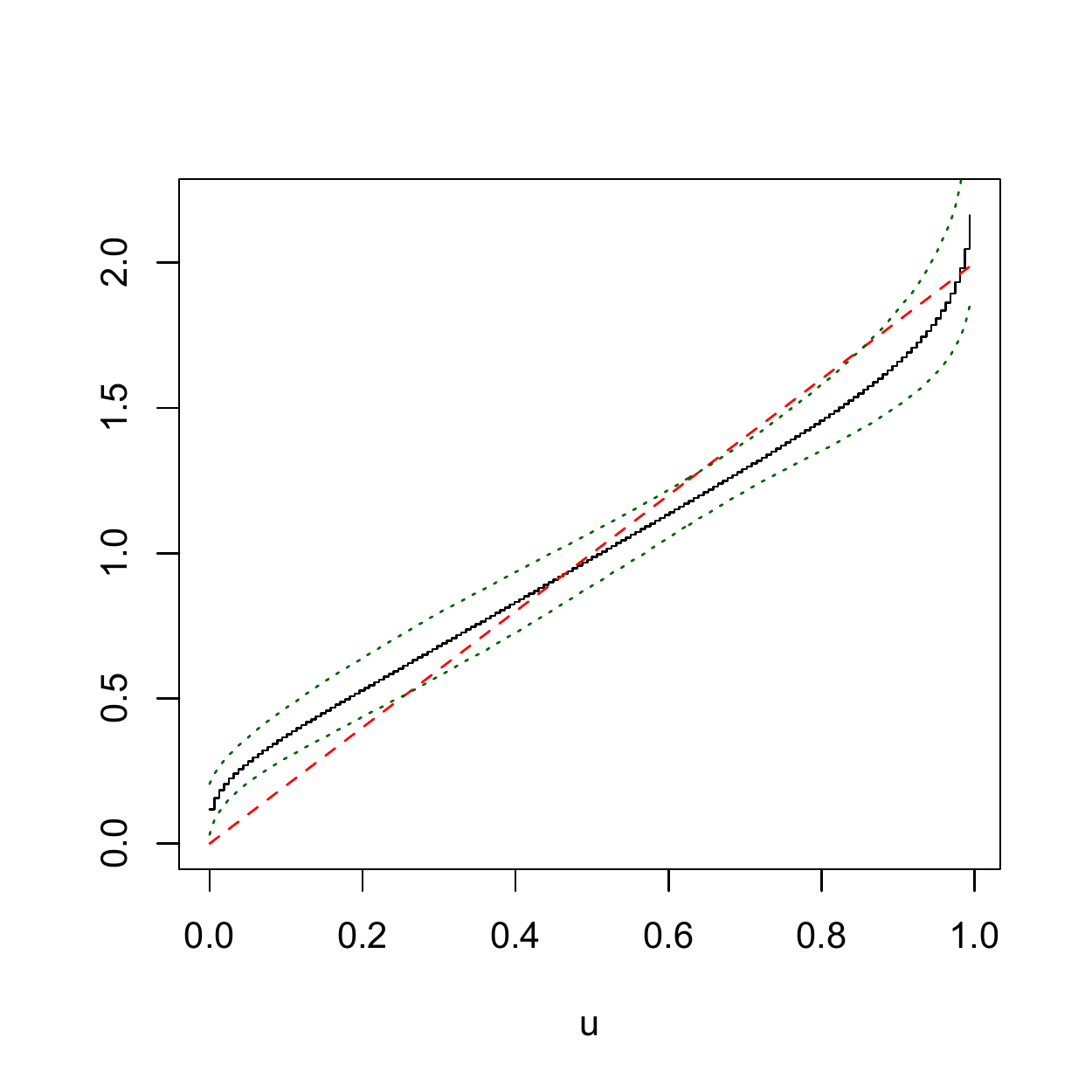} &
\includegraphics[height=7cm,width=6.5cm]{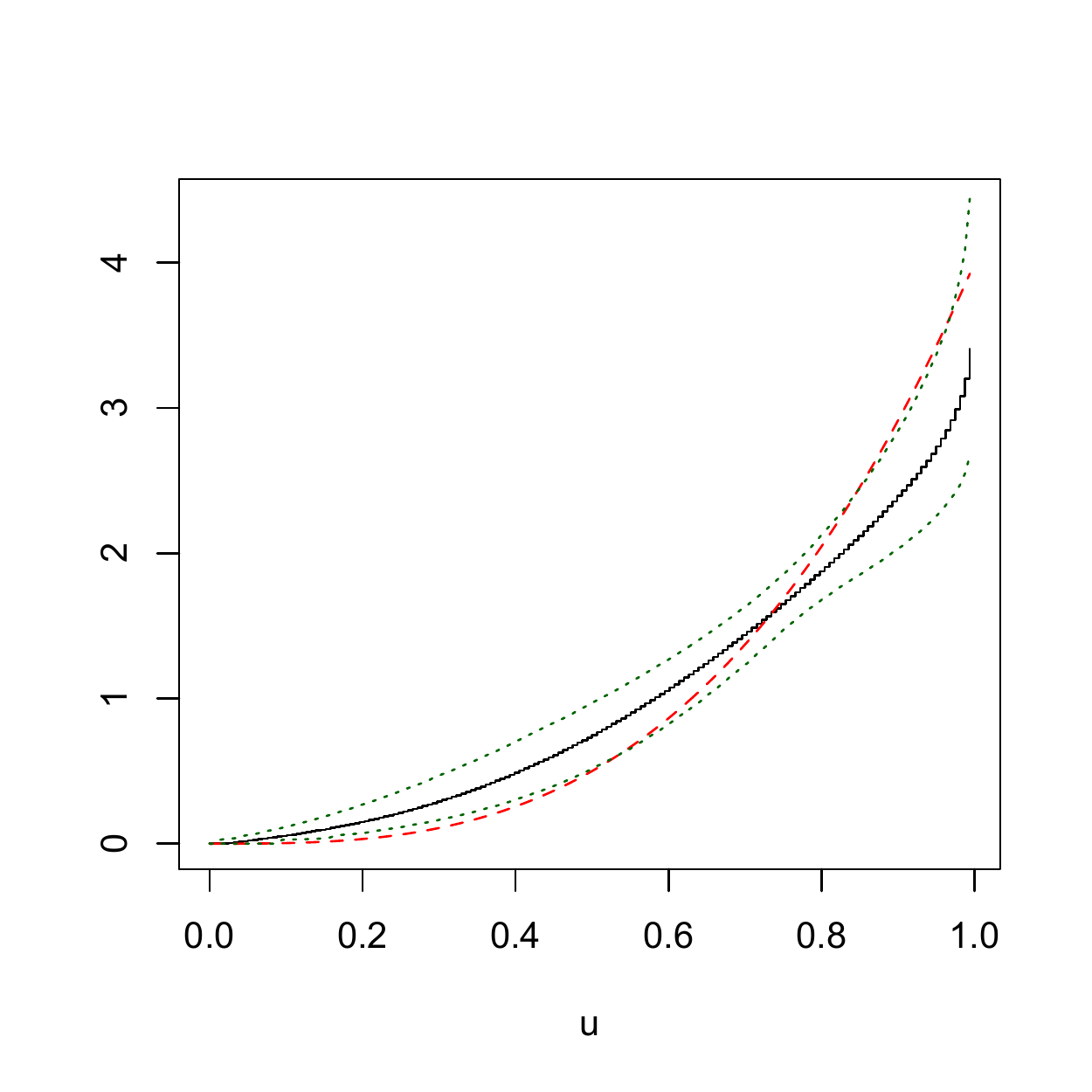}
\end{tabular}
\caption{\textit{10 000 Monte Carlo simulations for $h(u)=2u$ (left) and $h(u)=4u^3$ (right). In red the function $h$, in black its estimation, and in green the 95 $\%$ confidence intervals.}}
\end{center}

We see that in each case, the function $h$ is indeed fairly well estimated. There is of course a positive bias around zero, especially in the first case. This is in particular due to the fact that our estimators cannot be negative, whereas in both cases $h(0)=0$.

\section{Conclusion}\label{conclu}

In this paper, we take advantage of the dynamic information from the order flow in order to estimate an efficient price, whereas the classical
approach is to use the static information from the imbalance of the order book. This is done through a simple model which follows the principle that ``the price is where the volume is not". Indeed we assume that the intensity of the flow at the considered level of the order book (best bid or best ask) is a deterministic function (order flow response function) of the distance between this level and the efficient price. This function is estimated in a non parametric way and the estimation of the efficient price is derived from it. 

Remark that in our approach, the arrival of the orders at the best bid and best ask levels are connected through the influence of the efficient price. However, in the estimating procedures, the bid and ask sides are
treated separately. An interesting extension of this work, left for further research, would be to find a suitable way to combine these estimators, or even to consider a more intricate dependence structure between the order flows at the best bid and best ask levels.

\section{Proofs}\label{proofs}
All the proofs are done under H1 and H2. In the following, $c$ denotes a constant that may vary from line to line, and
even in the same line.

\subsection{\label{Proofmuchap}Proof of Proposition \ref{muchap}}
We give here an elementary proof of Proposition \ref{muchap}. 
However, note that this result can also be obtained as a by product of Corollary \ref{cortemp} shown in Section \ref{pt1t12}.

The characteristic function of $\sqrt{T}(\hat{\mu}_{T}/\mu
_{T}-1) $ satisfies 
$$\mathbb{E}\big[\exp\big(it\sqrt{T}( \hat{\mu}_{T}/\mu
_{T}-1)\big)\big]=\exp(-it\sqrt{T})\mathbb{E}\Big[\mathbb{E}_Y\big[
\exp\big(it\frac{N_{T}}{\mu _{T}\sqrt{T}}\big)\big]\Big],$$ 
where $\mathbb{E}_Y$ denotes the expectation conditional of the path of $(Y_t)$. Since conditional
on the path of $(Y_t)$, $N_{T}$ follows a Poisson random variable with parameter
$$\mu_T\int_{0}^Th(Y_t)dt,$$ the characteristic function is equal to
$$\exp(-it\sqrt{T})\mathbb{E}\Big[
\exp\Big((\text{e}^{it(\mu _{T}\sqrt{T})^{-1}}-1)\mu_T\int_{0}^Th(Y_t)dt\Big)\Big].$$
Now 
$$\big|\text{e}^{it(\mu _{T}\sqrt{T})^{-1}}-1-it(\mu  _{T}\sqrt{T})^{-1}\big|\leq t^2(\mu _{T}\sqrt{T})^{-2}.$$
Using the boundedness of $h$ and the fact that $\mu_T$ goes to infinity, we get that the limit as $T$ tends to infinity of the characteristic function is given by
$$\underset{T\rightarrow+\infty}{\text{lim}}\mathbb{E}\Big[
\exp\Big((it\sqrt{T})\big(\frac{1}{T}\int_{0}^Th(Y_t)dt-1\big)\Big)\Big].$$
From Equation \eqref{CLT1}, this is equal to 
$$\exp\big(-\frac{\sigma^2t^{2}}{2}\text{Var}[Z^h]\big).$$

\subsection{An auxiliary result}
We now give an auxiliary result from which Theorem \ref{thhinv} and Theorem \ref{TCL} will be easily deduced.
For $j=1,\ldots ,k_{T}$, we set
\begin{equation*}
\theta _{j}=\frac{k_{T}}{\mu _{T}T}(
N_{jT/k_{T}}-N_{(j-1)T/k_{T}}) =\frac{\hat{\mu}_{T}}{\mu _{T}}\hat{%
\theta}_{j},
\end{equation*}
and define the function $\hat{h}_{e}$ by%
\begin{equation*}
\hat{h}_{e}(u)=\theta _{(\lfloor uk_{T}\rfloor+1) }=%
\hat{h}(u)\frac{\hat{\mu}_{T}}{\mu _{T}},
\end{equation*}%
for $u\in \lbrack 0,1)$. Note that \begin{equation*}
\frac{\hat{\mu}_{T}}{\mu _{T}}=\frac{N_{T}}{T\mu _{T}}=\frac{1}{k_{T}}%
\sum_{j=1}^{k_{T}}\theta _{j}=\int_{0}^{1}\hat{h}_{e}(u)du.
\end{equation*}%
We also define 
\begin{equation*}
\hat{h}_{e}^{-1}(t )=\frac{1}{k_{T}}\sum_{j=1}^{k_{T}}\mathbb{I}_{\{
\theta _{j}\leq t\} },
\end{equation*}%
which satisfies 
\begin{equation*}
\hat{h}^{-1}(t )=\hat{h}_{e}^{-1}\big( t \frac{\hat{\mu}_{T}}{\mu
_{T}}\big).
\end{equation*}
We have the following proposition:
\begin{proposition}
\label{Propconvhemoinsun}As $T$ tends to infinity, we have
\begin{equation*}
\sqrt{T}\big(\hat{h}_{e}^{-1}( \cdot) -h^{-1}( \cdot
)\big) \overset{d}{\rightarrow} \sigma G( \cdot), 
\end{equation*}%
in $D[0,h(1^{-}))$.
\end{proposition}

By Theorem 15.1 in \cite{Bi68}, it is enough to prove that
the finite dimensional distributions converge and that a tightness criterion
holds. The proof is splitted into two Lemmas: Lemma \ref{Lemma1} for the
convergence of the finite dimensional distributions and Lemma \ref{Lemma2} for the
C-tightness. 

\begin{lemma}
\label{Lemma1}For any $n\in \mathbb{N}^{\ast }$, for any $( t
_{1},\ldots ,t_{n}) \in \lbrack 0,h(1^{-}))^{n}$, as $T\rightarrow
\infty $,
\begin{equation*}
\sqrt{T}\big( \hat{h}_{e}^{-1}(t _{i}) -h^{-1}(
t_{i})\big) _{i=1,\ldots ,n}\overset{d}{\rightarrow }N(
0,\sigma^2\nu) 
\end{equation*}%
where $\nu _{i,j}=\rho \left( t_{i},t_{j}\right)$, with $\rho$ the covariance function defined in \eqref{defrho}.
\end{lemma}

\begin{proof}
We give the proof for $n=1$. The other cases are obtained
the same way by linear combinations of the components of $\sqrt{T}\big( 
\hat{h}_{e}^{-1}( t_{i}) -h^{-1}( t_{i})
\big) _{i=1,\ldots ,n}$.

We write $\hat{h}_{e}^{-1}(t )-h^{-1}(t)=T_1+T_2+T_3$, with%
\begin{align*}
T_1=&\frac{1}{%
k_{T}}\sum_{j=1}^{k_{T}}\mathbb{I}_{\{ \theta _{j}\leq t\}
}-\frac{1}{k_{T}}\sum_{j=1}^{k_{T}}\mathbb{I}_{\{ \frac{k_{T}}{T}%
\int_{( j-1) T/k_{T}}^{jT/k_{T}}h( Y_{u}) du\leq
t \} }, \\
T_2=&\frac{1}{k_{T}}\sum_{j=1}^{k_{T}}\mathbb{I}_{\{ \frac{k_{T}}{T%
}\int_{( j-1) T/k_{T}}^{jT/k_{T}}h( Y_{u}) du\leq
t\} }-\frac{1}{T}\sum_{j=1}^{k_{T}}\int_{( j-1)
T/k_{T}}^{jT/k_{T}}\mathbb{I}_{\{ h( Y_{u}) \leq t
\} }du, \\
T_3=&\frac{1}{T}\int_{0}^{T}\mathbb{I}_{\{ h( Y_{u})
\leq t\} }du-h^{-1}( t).
\end{align*}
We study each component separately.
We have%
\begin{equation*}
T_1=\frac{1}{k_{T}}\sum_{j=1}^{k_{T}}v_{j,k_{T},T}
\end{equation*}%
with $v_{j,k_{T},T}=0$ or 

\qquad - $v_{j,k_{T},T}=1$ if 
\begin{equation*}
\theta _{j}\leq t <\frac{k_{T}}{T}\int_{( j-1)
T/k_{T}}^{jT/k_{T}}h( Y_{t}) dt,
\end{equation*}

\qquad - $v_{j,k_{T},T}=-1$ if%
\begin{equation*}
\frac{k_{T}}{T}\int_{( j-1) T/k_{T}}^{jT/k_{T}}h(
Y_{t}) dt\leq t <\theta _{j}.
\end{equation*}
For $\varepsilon_T>0$, we have $\mathbb{E}[|v_{j,k_{T},T}|]\leq A_1+A_2$, with
\begin{align*}
A_1&=\mathbb{P}\big[|\frac{k_T}{T}\int_{(j-1)T/k_T}^{jT/k_T}h(Y_u)du-t\big|\leq\varepsilon_T\big],\\
A_2&=\mathbb{E}\big[ 
\mathbb{I}_{\{ v_{j,k_{T},T}=\pm 1\} }\mathbb{I}_{\{
|t -\frac{k_{T}}{T}\int_{( j-1)
T/k_{T}}^{jT/k_{T}}h( Y_{t}) dt| >\varepsilon
_{T}\} }\big].
\end{align*}
We have
$$\big|\frac{k_T}{T}\int_{(j-1)T/k_T}^{jT/k_T}h(Y_t)dt-h(Y_{(j-1)T/k_T})\big|\leq c\frac{k_T}{T}\int_{(j-1)T/k_T}^{jT/k_T}|Y_t-Y_{(j-1)T/k_T}|dt.$$
Now remark that from Markov inequality and the fact that $Y_s$ is uniformly distributed, we easily get that for any $p>0$, $\delta>0$, $\lambda>0$:
$$\mathbb{P}\big[\underset{t\in[s,s+\delta]}{\text{sup}}|Y_t-Y_s|\geq \lambda\big]\leq c_p(\delta^{p/2}\lambda^{-p}+\delta^{1/2}).$$
Therefore, for $p>0$, 
$$\mathbb{P}\big[|\frac{k_T}{T}\int_{(j-1)T/k_T}^{jT/k_T}h(Y_t)dt-h(Y_{(j-1)T/k_T})|\geq\varepsilon_T\big]\leq c\frac{(T/k_T)^{p/2}}{\varepsilon_T^p}+c(T/k_T)^{1/2}.$$
Thus
$$A_1\leq\mathbb{P}\big[|h(Y_{(j-1)T/k_T})-t\big|\leq2\varepsilon_T\big]+c\frac{(T/k_T)^{p/2}}{\varepsilon_T^p}+c(T/k_T)^{1/2}.$$
Using the stationary distribution of $Y_{(j-1)T/k_T}$ we obtain that for $\varepsilon_T$ small enough,  
$$A_1\leq h^{-1}(t+2\varepsilon_T)-h^{-1}(t-2\varepsilon_T)+c\frac{(T/k_T)^{p/2}}{\varepsilon_T^p}+c(T/k_T)^{1/2}\leq c\varepsilon_T+c\frac{(T/k_T)^{p/2}}{\varepsilon_T^p}+c(T/k_T)^{1/2}.$$

We now turn to $A_2$. Remarking that $|v_{j,k_{T},T}|=1$ implies
$$\big|\theta_j-\frac{k_{T}}{T}\int_{( j-1) T/k_{T}}^{jT/k_{T}}h(
Y_{t}) dt\big|\geq \big|t-\frac{k_{T}}{T}\int_{( j-1) T/k_{T}}^{jT/k_{T}}h(
Y_{t}) dt\big|,$$ we get
$$A_2\leq \mathbb{P}\big[\big|\theta_j-\frac{k_{T}}{T}\int_{( j-1) T/k_{T}}^{jT/k_{T}}h(
Y_{t}) dt\big|> \varepsilon
_{T}\big].$$
Recall that conditional on the path of $(Y_t)$, $\theta _{j}\overset{d}{=}Zk_T/(\mu _{T}T)$, where $Z$ is a Poisson random variable with parameter
$$\mu
_{T}\int_{(j-1)T/k_{T}}^{jT/k_{T}}h(Y_{t})dt.
$$
Thus we obtain
$$A_2\leq c\frac{k_{T}}{\mu_TT\varepsilon_T^2}.$$
Using that \begin{equation*}
\mathbb{E}[|T_1|]\leq \frac{1}{k_{T}}\sum_{j=1}^{k_{T}}%
\mathbb{E}[|v_{j,k_{T},T}|],
\end{equation*}%
we finally get
$$\mathbb{E}[|\sqrt{T}T_1|]\leq c\varepsilon_T\sqrt{T}+c\frac{T^{(p+1)/2}}{\varepsilon_T^pk_T^{p/2}}+c\frac{k_{T}}{\mu_T\sqrt{T}\varepsilon_T^2}+c\frac{T}{\sqrt{k_T}}.$$
We take $\varepsilon_T=\zeta_T/\sqrt{T}$, with $\zeta_T$ tending to zero. Then
$$\mathbb{E}[|\sqrt{T}T_1|]\leq c\zeta_T+c\frac{T^{p+1/2}}{\zeta_T^pk_T^{p/2}}+c\frac{k_{T}\sqrt{T}}{\mu_T\zeta_T^2}+c\frac{T}{\sqrt{k_T}}.$$
Thanks to the assumptions on $k_T$, we can find a sequence $\zeta_T$ tending to $0$ such that
$\mathbb{E}[|\sqrt{T}T_1|]$ goes to $0$.

We now turn to $T_2$. We have 
\begin{align*}
T_2&=\frac{1}{T}\sum_{j=1}^{k_{T}}\int_{( j-1)
T/k_{T}}^{jT/k_{T}}\big(\mathbb{I}_{\{ \frac{k_{T}}{T}\int_{(
j-1) T/k_{T}}^{jT/k_{T}}h( Y_{u}) du\leq t\} }-%
\mathbb{I}_{\{ h( Y_{s}) \leq t\} }\big) ds\\
&=\frac{1}{T}\sum_{j=1}^{k_{T}}\int_{( j-1)
T/k_{T}}^{jT/k_{T}}w_{j,k_{T},s,T}ds
\end{align*}
with $w_{j,k_{T},s,T}=0$ or 

\qquad - $w_{j,k_{T},s,T}=1$ if 
\begin{equation*}
\frac{k_{T}}{T}\int_{( j-1) T/k_{T}}^{jT/k_{T}}h(
Y_{u}) du\leq t <h( Y_{s})
\end{equation*}

\qquad - $w_{j,k_{T},t,T}=-1$ if%
\begin{equation*}
h( Y_{s}) \leq t <\frac{k_{T}}{T}\int_{( j-1)
T/k_{T}}^{jT/k_{T}}h(Y_{u}) du.
\end{equation*}

Now remark that $|w_{j,k_{T},s,T}|=1$ implies
$$\big|h(Y_s)-\frac{k_{T}}{T}\int_{( j-1) T/k_{T}}^{jT/k_{T}}h(
Y_{u}) du\big|\geq \big|t-\frac{k_{T}}{T}\int_{( j-1) T/k_{T}}^{jT/k_{T}}h(
Y_{u}) du\big|.$$
Using the same kind of computations as for $T_1$, for $\varepsilon_T'$ positive and small enough and $p>0$, this leads to
$$\mathbb{P}\big[|w_{j,k_{T},s,T}|=1\big]\leq c\varepsilon'_T+ c\frac{(T/k_T)^{p/2}}{(\varepsilon'_T)^p}+c(T/k_T)^{1/2}.$$
Then, in the same way as for $T_1$, we easily obtain that $\mathbb{E}[|\sqrt{T}T_2|]$ goes to $0$.

Eventually, from \eqref{CLT1}, we get that $\sqrt{T}T_3$ converges in law towards a centered Gaussian random variable with variance $\sigma^2\text{Var}[Z_{e}(t)]$.

\end{proof}

We now give the second Lemma needed to prove Proposition \ref{Propconvhemoinsun}.

\begin{lemma}
\label{Lemma2}Let%
\begin{equation*}
\alpha _{T}(t) =\sqrt{T}\big( \frac{1}{T}\int_{0}^{T}%
\mathbb{I}_{\{ h( Y_{s}) \leq t\}
}ds-h^{-1}(t) \big) .
\end{equation*}%
The sequence $\big( \alpha _{T}( t)\big) _{0\leq t
<h(1^{-})}$ is tight in $D[0,h(1^{-}))$.
\end{lemma}

\begin{proof} Recall first the following classical C-tightness criterion, see Theorem 15.6 in \cite{Bi68}: 
If for $p>0$, $p_1>1$ and all $0\leq t_{1}, t_{2}<h(1^{-})$%
\begin{equation*}
\mathbb{E}\big[|\alpha _{T}( t_{1}) -\alpha _{T}(
t _{2})|^{p}\big]\leq c\big(|t
_{1}-t _{2}|^{p_{1}}\big), 
\end{equation*}%
then $\big(\alpha _{T}(
t)\big) _{0\leq t \leq h(1)}$ is tight in $D[0,h(1^{-})).$

We now use the sequence of stopping time $\left( \nu _{n}\right) _{n\in 
\mathbb{N}}$ defined by \eqref{stop}. 
Let $$n_{T}=\inf\{
i:\nu _{i}\geq T\}.$$ From Theorem II.5.1 and Theorem II.5.2 in \cite{Gu88}, we have
\begin{equation}\label{nm1}
\E[\nu_{n_{T}}]\leq cT
\end{equation}
and 
\begin{equation}\label{nm2}
\E[(\frac{n_T}{\sigma^2}-T)^2]\leq cT.
\end{equation}

We now define
\begin{equation*}
Y_{i}( t _{1},t _{2}) =\frac{1}{\sqrt{T}}\Big(
\int_{\nu _{i-1}}^{\nu _{i}}\mathbb{I}_{\{ t _{1}<h(
Y_{t}) \leq t _{2}\} }dt-\frac{1}{\sigma^{2}}\big(h^{-1}( t
_{2}) -h^{-1}( t _{1})\big) \Big) 
\end{equation*}%
and%
\begin{equation*}
\tilde{Y}_{n_{T}+1}( t _{1},t _{2}) =\frac{1}{\sqrt{T}}
\int_T^{\nu _{n_{T}}}\mathbb{I}_{\{ t _{1}<h(
Y_{t}) \leq t _{2}\}}dt. 
\end{equation*}
Remark that for $i\geq 2$, the $Y_{i}( t _{1},t _{2})$ are centered (see Section \ref{reg}) and iid.
Moreover, using the occupation formula together with a Taylor expansion, we get
$$
\mathbb{E}\big[\big(Y_{2}( t _{1},t _{2})\big)^2\big]\leq T^{-1}\E\Big[\big(\int_{\nu
_{i-1}}^{\nu _{i}}\mathbb{I}_{\{ t _{1}<h( Y_{t})
\leq t _{2}\} }dt\big)^2\Big]
\leq cT^{-1}|t_2-t_1|^2\E[(L^*)^2],$$
with $$L^*=\underset{u\in[-1/\sigma,1/\sigma]}{\text{sup}}\big(L_{-1/\sigma,1/\sigma}(u)\big).$$
From the Ray-Knight version of Burkh\"older-Davis-Gundy inequality, see \cite{ry}, we know
that all polynomial moments of $L^*$ are finite and thus
$$\mathbb{E}\big[\big(Y_{2}( t _{1},t _{2})\big)^2\big]\leq cT^{-1}|t_2-t_1|^2.$$
We show in the same way that
$$\mathbb{E}\big[\big(Y_{1}( t _{1},t _{2})\big)^2\big]+\mathbb{E}\big[\big(\tilde Y_{n_T+1}( t _{1},t _{2})\big)^2\big]\leq cT^{-1}|t_2-t_1|^2.$$

We now use the preceding inequalities in order to show that the tightness criterion holds. We have $$\alpha _{T}(t _{1}) -\alpha
_{T}( t _{2})=B_1+B_2-B_3+B_4,$$ with
\begin{align*}
B_1&=Y_{1}( t _{1},t _{2}),~~
B_2=\sum_{i=2}^{n_{T}}Y_{i}( t _{1},t _{2}),~~
B_3=\tilde Y_{n_{T}+1}( t _{1},t _{2}),\\
B_4&=\sqrt{T}\big( \frac{n_{T}-1%
}{\sigma^2T}-1\big)\big( h^{-1}( t _{2}) -h^{-1}( t
_{1}) \big).
\end{align*}
Thus
$$\E\big[|\alpha _{T}(t _{1}) -\alpha
_{T}( t _{2})|^2\big]\leq c\sum_{i=1}^4\E[|B_i|^2].$$
 By Theorem I.5.1 in \cite{Gu88}, 
\begin{equation*}
\mathbb{E}\big[\big|\sum_{i=2}^{n_{T}}Y_{i}( t _{1},t
_{2})\big|^{2}\big]\leq c\mathbb{E}[n_{T}] 
\mathbb{E}\big[\big(Y_{2}( t _{1},t _{2})\big)
^{2}\big]. 
\end{equation*}
Moreover, using \eqref{nm1}, we obtain
\begin{equation*}
\mathbb{E}[n_{T}] 
\mathbb{E}\big[|Y_{2}( t _{1},t _{2})|
^{2}\big]\leq c|t _{1}-t
_{2}|^{2}.
\end{equation*}
Finally, we easily get from \eqref{nm2} that 
$$\E[(B_4)^2]\leq c|t _{1}-t
_{2}|^{2}.$$
and therefore the tightness criterion is satisfied.
\end{proof}

\subsection{Proofs of Theorem \ref{thhinv} and Theorem \ref{TCL}}\label{pt1t12}

In order to prove Theorem \ref{thhinv} and Theorem \ref{TCL}, we start with the following corollary of Proposition \ref{Propconvhemoinsun}:

\begin{corollary}\label{cortemp}
As $%
T\rightarrow \infty $, 
$$
\sqrt{T}\Big(\hat{h}^{-1}_{e}(\cdot) -h^{-1}(\cdot),\hat{h}_{e}(\cdot) -h(\cdot)
,\int_{0}^{1}\hat{h}_{e}(u)du-1\Big)$$ 
converges in law towards
$$
\Big(\sigma G(\cdot), -h^{\prime
}(\cdot )\sigma G\big(h(\cdot )\big) ,-\int_{0}^{h(1^{-})}\sigma G( v)
dv\Big), 
$$ in $D[0,h(1^{-}))\times D[0,1)\times\mathbb{R}$ (for the product topology).
\end{corollary}

\begin{proof}
The result follows directly from Proposition \ref{Propconvhemoinsun} together with Theorem 13.7.2 in \cite{Wh02}
and the continuous mapping theorem. 
\end{proof}

We now give the proof of Theorem \ref{thhinv}. We write
$$\sqrt{T}\big(\hat{h}^{-1}( \cdot) -h^{-1}(\cdot)
\big)=V_1+V_2,$$ with
\begin{align*}
V_1&=\sqrt{T}\big(( \hat{h}_{e}^{-1}( \cdot \hat{\mu}
_{T}/\mu _{T})-h^{-1}(\cdot\hat{\mu}_{T}/\mu _{T})\big),\\
V_2&=\sqrt{T}\big((h^{-1}( \cdot \hat{\mu}
_{T}/\mu _{T})-h^{-1}(\cdot)\big).
\end{align*}
From Proposition \ref{Propconvhemoinsun} together with the continuity of the composition map, see Theorem 13.2.2 in \cite{Wh02}, we have
\begin{equation*}
\sqrt{T}V_1 \overset{d}{
\rightarrow} \sigma G(\cdot),
\end{equation*}%
in $D[0,h(1^{-}))$. 
Now recall that
$$\int_{0}^{1}\hat{h}_{e}(u)du=\hat{\mu}_{T}/\mu _{T}.$$
Thus using Corollary \ref{cortemp} together with Theorem 13.3.3 in \cite{Wh02}, we get
\begin{equation*}
\sqrt{T}V_2 \overset{d}{
\rightarrow} -\frac{(\cdot )}{h^{\prime }\big( h^{-1}(\cdot )\big) }%
\int_{0}^{h(1^{-})}\sigma G(v) dv,
\end{equation*}%
in $D[0,h(1^{-}))$. The convergences of $\sqrt{T}V_1$ and $\sqrt{T}V_2$ taking place jointly, this ends the proof of Theorem \ref{thhinv}.

From Theorem \ref{thhinv}, Theorem \ref{TCL} follows from Theorem 13.7.2 in \cite{Wh02}.

\subsection{Proof of Corollary \ref{effestim}}
Let $\widehat{h_e(Y_t)}$ be the oracle estimate of $h(Y_t)$ defined the same way as $\widehat{h(Y_t)}$ but with
$\mu_T$ instead of $\hat\mu_T$.
Now we use that  $$\hat{h}^{-1}\big(\cdot\frac{\mu _{T}}{\hat{\mu}_{T}}\big)=\hat{h}_e^{-1}(\cdot)$$
and the fact that the cycles on which $\widehat{h_e(Y_t)}$ is defined are negligible in the estimation of 
$h^{-1}( \cdot)$ and $\mu_T$. Then, from Proposition \ref{Propconvhemoinsun}, we have
$$\Big(\sqrt{T}\big( \hat{h}^{-1}\big(\cdot\frac{\mu _{T}}{\hat{\mu}_{T}}\big)
-h^{-1}( \cdot) \big),\widehat{h_e(Y_t)}\Big)\overset{d}{
\rightarrow}\big(\sigma G(\cdot),h(Y_t)\big)$$ in $D[0,h(1^{-}))\times \mathbb{R}^+$, with $G$ independent of $Y_t$.
Using Theorem 13.2.1 in \cite{Wh02} and the fact that 
$$\widehat{h(Y_t)}=\frac{\mu _{T}}{\hat{\mu}_{T}}\widehat{h_e(Y_t)},$$
we derive
$$\sqrt{T}\Big(\hat{h}^{-1}\big(\widehat{h(Y_t)}\big)
-h^{-1}\big(\widehat{h_e(Y_t)}\big)\Big)\overset{d}{
\rightarrow}\sigma G\big(h(Y_t)\big).$$
Now recall that
$$\big|h^{-1}\big(\widehat{h_e(Y_t)}\big)-h^{-1}\big(h(Y_t)\big)\big|\leq c\big|\widehat{h_e(Y_t)}-h(Y_t)\big|.$$
and that conditional on the path of $(Y_t)$, $\widehat{h_e(Y_t)}\overset{d}{=}Zk_T/(\mu _{T}T)$, where $Z$ is a Poisson random variable with parameter
$$\mu
_{T}\int_{t-T/k_{T}}^{t}h(Y_{u})du.
$$
Therefore, using a Taylor expansion and the fact that
$$\mathbb{P}\big[\underset{u\in [t-T/k_{T},t]}{\text{sup}}|Y_u-Y_t|\geq 1\big]\leq c(T/k_{T})^{1/2},$$
we get that $$\E\Big[\big|\widehat{h_e(Y_t)}-h(Y_t)\big|\Big]$$ is smaller than
$$c\E\big[\underset{u\in [t-T/k_{T},t]}{\text{sup}}|Y_u-Y_t|\big]+c\big(k_T/(T\mu_{T})\big)^{1/2}\leq c(T/k_T)^{1/2}+c\big(k_T/(T\mu_{T})\big)^{1/2}.$$
Eventually, since $T/\sqrt{k_T}$ and $k_T\sqrt{T}/\mu_T$ tend to zero, we get
$$\sqrt{T}\E\Big[\big|\widehat{h_e(Y_t)}-h(Y_t)\big|\Big]\leq cT/\sqrt{k_T}+c(k_T/\mu_{T})^{1/2}\rightarrow 0,$$ which concludes the proof.

\section*{Acknowledgements} We are grateful to Renaud Drappier and Sebouh Takvorian from BNP Paribas for their very interesting comments. 
We also thank the referee for his helpful remarks.

\end{document}